\documentclass{article}
\usepackage{amsmath, amssymb,amsthm}

\theoremstyle{theorem}

\newtheorem{proposition}{Proposition}
\newtheorem*{lemma*}{Lemma}

\theoremstyle{definition}
\newtheorem*{definition*}{Definition}
\newtheorem*{remark*}{Remark}
\newtheorem*{example*}{Example}
\usepackage{mathtools}
\usepackage{subcaption}

\begin{document}

\title{A Mechanism Design Approach to Allocating Travel Funds}

\author{Michael A. Jones} 

\maketitle



\begin{abstract} 
I explain how faculty members could exploit a method to allocate travel funds and how to use game theory to design a method that cannot be manipulated.    \end{abstract}
\bigskip

In mathematics and other disciplines, faculty members are required to give talks at conferences on their research or teaching.  While I was a professor at Montclair State University in New Jersey, the financial requests for travel exceeded the amount the Dean's office for the College of Science and Mathematics had budgeted, which meant that only a portion of the travel costs was covered.  Whether on purpose or not, faculty could present inaccurate information about their travel needs and expenses and have a higher portion of their costs paid.  But, because the amount to be allocated was fixed, this meant that other faculty would receive less money.  The potential for exploiting the system led the Associate Dean to ask me to construct a new method to allocate travel funds that could not be exploited. Because it was possible for faculty to game the system, it seemed reasonable to use game theory to construct this new method.
 
In this paper, I explain the old method to award travel funds and how faculty could misrepresent their financial needs to have a higher portion of their travel paid.  For the new method, it is in each of the faculty member's best interest to reveal truthfully his or her financial needs. The method has the added benefit that it encourages faculty to be conservative in their spending so that each faculty member has a higher portion of their travel paid.  After an optimization perspective, I explain how awarding travel funds may be viewed as a game for which it is a Nash equilibrium for each faculty member to reveal truthfully how much money he or she needs to travel.  Constructing a game to achieve a particular outcome, in this case the truthful revelation of the financial requirements for faculty travel, is called mechanism design, a topic for which the Nobel Prize in Economics was awarded in 2007 \cite{Nobel3}.

\section*{The Original Method and Its Exploitation}

At the start of each academic year, the Dean's office asks each faculty member to submit a financial request indicating how much money he or she needs to attend one conference. This may seem difficult at first, because a faculty member may wish to attend a conference in April, but has to know how much he or she will spend before spending it!  Because registration fees and hotel accommodations for conferences are posted well in advance, faculty either estimate their airfare or book their airfare in advance, too.  (It is rare for faculty to submit a request for travel and then decide not to attend because the Dean's office doesn't offer enough support.  For this reason, faculty tend to book their flights in advance, too.)  At [Blinded University], faculty receive a per diem, a fixed amount of money for meals each day.  This means that how much faculty members need for travel can be determined before submitting their travel requests (as long as they remember to include taxes for the hotel room and any costs incurred to and from the airport).  Based on these requests, the Dean's office sets aside or budgets an amount for each faculty member to spend on travel.  A faculty member could ask for a cash advance for the amount set aside, but more often than not, faculty members would return from the conference and return receipts to recoup the budgeted amount.

To model the travel funds problem, let $r_i$ be the amount (in dollars) requested by faculty member $i$.  Let $R$ be the total amount requested by the $n$ faculty members so that $R = r_1 + \cdots + r_n$ and let $B$ be the amount the Dean's office has budgeted for travel.  If $B \ge R$, then faculty members are budgeted the amounts they requested.  The interesting question is how to budget money for each faculty member's travel when $B < R$.  Under the old method, the Dean's office budgets $b_i = (B/R)r_i$ for faculty member $i$'s travel, so that each faculty member receives the same  proportion $B/R$ of the amount requested. Consequently, all funds are budgeted because $B = b_1 + \cdots + b_n$.  Faculty member $i$ is responsible for any amount spent above $(B/R)r_i$ and only needs to supply receipts for the allocated amount $(B/R)r_i$ because the Dean's office will not reimburse more. The following example demonstrates the old method of allocating travel funds.

\begin{example*}  Suppose that the total requested by the faculty is $R = \text{100,000}$, but the Dean's office has budgeted only $B= \text{75,000}$ for travel.  Under the old method, then each faculty member is budgeted $3/4$ ($=\text{75,000}/\text{100,000}$) of the amount he or she requests.  If faculty member $i$ requests $r_i = \text{1200}$, then she is budgeted $b_i =  (3/4)\times \text{1200} = 900$. \end{example*}

Because faculty member $i$ has no control over the other faculty members' requests, define $R_{-i} = R - r_i$ to highlight the other faculty members' contributions to $R$.  Faculty member $i$ can view the amount he or she is budgeted as a real-valued function of $\mathbf r = (r_1, r_2, \ldots, r_n)$ in which
\begin{align*}
f_i (\mathbf r) &= (B/R) r_i = \left( \frac{B}{R_{-i} + r_i} \right) r_i  = \left( \frac{B}{R_{-i} + r_i} \right) \left( R_{-i} + r_i \right) - \frac{B R_{-i}}{R_{-i} + r_i} \\
&  = B - \frac{B R_{-i}}{R_{-i} + r_i} = B - BR_{-i} \left( R_{-i} + r_i \right)^{-1}.
\end{align*}
Notice that $f_i$ is increasing in $r_i$ because $\partial f_i(\mathbf r)/ \partial r_i = BR_{-i} \left( R_{-i} + r_i \right)^{-2} > 0$ for $r_i, R_{-i} > 0$.  This makes sense, as a faculty member is budgeted more if he or she requests more money.  It is this property that led faculty members at [Blinded University] to game the system.  Let's return to the earlier example to see how this works.  Indeed, if $R_{-i}$ and $B$ are known, then faculty member $i$ could request an amount to have his or her entire trip paid for!

\begin{example*}[continued] From before, $B = \text{75,000}$ and $R_{-i} = \text{98,800}$.  Suppose that faculty member $i$ really does need 1200 dollars to travel to a conference.  Then, faculty member $i$ can request $r_i^*$ to be budgeted $b_i = 1200$ by solving for $r_i^*$ in 
\begin{equation*} 
1200 = f_i((r_i^*, \mathbf r_{-i})) = \left( \frac{B}{R_{-i} + r_i^*} \right) r_i^* \text{ so that } r_i^* =  \frac{1200 R_{-i}}{B - 1200} \approx 1606.50,
\end{equation*}
where $(r_i^*, \mathbf r_{-i}) = (r_1, \ldots, r_{i-1},r^*_i, r_{i+1}, \ldots, r_n)$ replaces $r_i$ in $\mathbf r$ with $r^*_i$.  Faculty member $i$ can request 1606.50, be budgeted 1200 from the Dean's office, spend 1200, and be reimbursed for the entire 1200. \end{example*}

In practice, faculty members know neither how much the Dean's office has to spend on travel nor how much the other faculty will request, but I assume that a faculty member is interested in spending as little money out of pocket as possible.  Because $f_i$ is increasing in $r_i$, the faculty member can always inflate how much he or she needs to travel and be budgeted more by the Dean's office, and therefore spend less out of pocket.  Every faculty member has this incentive to misrepresent his or her needs.  But, if every faculty member requests twice as much as they need, then they are budgeted the same amount as if each faculty member asked for the honest amount needed. This could lead to a travel request race, like an arms race, in which faculty ask for more and more money for their one trip, even if they could never spend the amount!

In an effort to keep faculty from misrepresenting how much they needed for their travels, the Dean's office required each faculty member $i$ to submit details of how the request $r_i$ would be spent.  A faculty member $i$ would provide information about hotel costs, registration fees, airfare, etc.  Yet, faculty members were still able to misrepresent the amount needed for travel.  Here is how they could do it.  

When attending a conference, there are usually many hotels affiliated with the conference, and other hotels close to the conference that are not affiliated with the conference. This mix of hotels includes a range of prices.  A faculty member could provide information to indicate that he or she would stay in a single occupant room at a more expensive hotel to justify a higher $r_i$. However, once the Dean's office computes $b_i$, faculty member $i$ wants to minimize how much is spent out of pocket by spending as close to $b_i$ as possible.  This could be done, whether intentionally gaming the system or not, by deciding to share a room at a less expensive hotel.

At this point, hopefully the problem is clear.  The Dean's office budgets the amounts $b_i$ based on $\mathbf r$ and $B$.  Once $b_i$ is determined, faculty member $i$ needs only to submit receipts up to $b_i$ to be reimbursed---meaning that a shrewd faculty member could spend less than $b_i$ and be reimbursed the full amount of the trip!  Even faculty who were not gaming the system would submit receipts up to $b_i$ and, as a way to cut down on paperwork, just not bother to collect receipts for expenses over $b_i$; this made it difficult to determine who was exaggerating travel requests.  (I always submitted all receipts to encourage the university to set aside more money for faculty travel the next year.)  The new method bases the amount budgeted to faculty member $i$ on $\mathbf r$ and $B$ proportionally as before, but reimburses according to how much is spent.  This slight change is enough to induce faculty to be truthful about how much they need for their travels.

\section*{New Method to Award Travel Funds and Truthful Revelation}

Let $s_i$ be the amount spent by faculty member $i$.  Under the old method, if each faculty member $i$ spends $s_i = r_i$, then faculty member $i$ is reimbursed the same proportion $b_i / s_i = b_i / r_i = B/R$ of the amount spent. If $b_i < s_i < r_i$, then faculty member $i$ is paid a higher proportion of his or her travel costs because $b_i/s_i > b_i/r_i$.  If $s_i \le b_i$, then faculty member $i$ is reimbursed 100\% of the travel costs.  (Of course, if faculty member $i$ spends more than requested, the proportion of his or her travel costs would be less than $b_i/r_i$.)  Besides the potential for exploiting the method to award travel funds, the old method was not fair because the proportions $b_i / s_i$ were not equal. 

Under the new method, when $B < R$, faculty member $i$ is still budgeted $b_i = (B/R) r_i$, but now the amount reimbursed depends on $s_i$ and $\mathbf r$ and is given by 
\begin{equation} \label{newreimbursementfunction}
g_i(\mathbf r, s_i)  = \left\{ \begin{array}{ll}
(B/R) s_i & \text{if $s_i \le r_i$} \\
(B/R) r_i & \text{if $s_i > r_i$.}
\end{array} \right.
\end{equation}
The method gives every faculty member the same proportion of the amount spent, as long as a faculty member spends up to the amount requested.  The Dean's office's hands are tied if a faculty member spends more than requested, because awarding the same proportion in this case could result in the Dean's office not having enough money to pay what was budgeted for another faculty member's travel.  Because faculty are able to determine how much they will spend for travel, the new method remedies the concern that faculty will be reimbursed different proportions.  Perhaps surprisingly, the simple change in how the travel money is reimbursed, according to both the amount spent and the amount requested, also addresses the Associate Dean's concern that faculty could exploit the method.  As before, we assume that faculty member $i$ tries to minimize her out-of-pocket expense, which is given by $s_i - g_i (\mathbf r, s_i)$, where $r_i$ is a variable and $s_i$ and $r_j$, for $j \ne i$, are fixed.

\begin{proposition} \label{prop1} To minimize out-of-pocket expenses, it is in the best interest of faculty member $i$ to request $r_i = s_i$.
\end{proposition}
\begin{proof}
Recall that $B < R$.  If faculty member $i$ is sincere in her request, by requesting $r_i = s_i$, then her out-of-pocket expense is $s_i - g_i(\mathbf r, s_i) = r_i - (B/R)r_i = (1 - B/R)r_i$.  The goal is to show that requesting $r_i \ne s_i$ results in a higher out-of-pocket expense.

If faculty member $i$ requests $r_i > s_i$, then her out-of-pocket expense is $s_i - (B/R)s_i$, which is greater than her out-of-pocket expense $s_i - (B / (s_i + R_{-i}))s_i$ if she had requested $s_i$; this follows because $B/(s_i + R_{-i}) > B/(r_i + R_{-i}) = B/R$.

If faculty member $i$  requests $r_i < s_i$, then her out-of-pocket expense is $s_i - (B/R)r_i$ because she is capped at the amount that she can be reimbursed, receiving the proportion $B/R$ of the amount she requested.  She pays more out of pocket by requesting $r_i < s_i$ if
\begin{equation*}
s_i - \left( \frac{B}{R} \right) r_i = s_i - \left( \frac{B}{r_i + R_{-i}} \right) r_i > s_i - \left( \frac{B}{s_i + R_{-i}} \right) s_i,
\end{equation*}
which is equivalent to 
\begin{equation*}
\left( \frac{B}{r_i + R_{-i}} \right) < \left(  \frac{B}{s_i + R_{-i}} \right).
\end{equation*} This last inequality is equivalent to the easy-to-verify inequality
$Bs_ir_i + Bs_i R_{-i} > Bs_i r_i + B r_i R_{-i}$ after cross multiplication.

Faculty member $i$ always pays less out of pocket by requesting $s_i$, the amount that she will spend. \end{proof}

The slight perturbation in how the reimbursements were awarded had an added benefit.  It encourages faculty to spend less, because decreasing how much they  spend decreases their out-of-pocket expenses. Spending less could be achieved by sharing a room or staying at less expensive hotel.

\begin{proposition}  To minimize out-of-pocket expenses, it is in the best interest of faculty member $i$ to spend less. 
\end{proposition}
\begin{proof}
Suppose that faculty member $i$ could spend $s_i^-$ or $s_i^+$ to attend the same conference, where $s_i^- < s_i^+$.  By Proposition \ref{prop1}, faculty member $i$ will request either $s_i^-$ or $s_i^+$. The amount spent out of pocket if faculty member $i$ requests and spends $s_i^-$ is $s_i^- ( 1- B/(s_i^- + R_{-i}))$.  Similarly, if faculty member $i$ requests and spends $s_i^+$, then her out-of-pocket expense is $s_i^+ (1 - B /(s_i^+ + R_{-i}))$.  When the faculty member requests and spends $s_i^-$, she is paid a higher proportion of her expenses than when she requests and spends $s_i^+$.  Because she is also spending less, then she is also paying less out of pocket. \end{proof}

\section*{Travel Fund Allocation Game and Mechanism Design}  

Up to this point, there has been no mention of a game because a faculty member's optimal behavior is independent of how the other faculty members behave.  Yet, this independence may be viewed from a game theory perspective.  To view the travel allocation problem as a game, first we consider each faculty member $i$, for $i = 1$ to $n$, as a player of the game.  Player $i$'s strategy is to submit a monetary request $r_i$, under the assumption that player $i$ can determine $s_i$.  As before, let $\mathbf r = (r_1, \ldots, r_n)$; $\mathbf r$ is called a strategy profile because it contains each player's strategy.  Define player $i$'s utility function to be $u_i (\mathbf r,s_i) = g_i(\mathbf r, s_i) - s_i$, where $g_i$ is the new reimbursement function from equation (\ref{newreimbursementfunction}).  Notice that $u_i(\mathbf r, s_i)$ is $(-1)$ times her out-of-pocket expenses, which is always less than or equal to zero.  In game theory, each player's goal is to maximize her utility function.  For this game, maximizing utility is the same as minimizing her out-of-pocket expenses.  

The most well-known solution concept in game theory is the Nash equilibrium.  As the term equilibrium suggests, a strategy profile $\mathbf r^*$ is a Nash equilibrium if it satisfies a notion of balance or stability in which no player has an incentive to change his or her strategy $r_i^*$, if every other player $j$ plays their strategy $r_j^*$.  The formal definition follows, where as before $(r_i, \mathbf r^*_{-i}) = (r^*_1, r_2^*, \ldots, r_{i-1}^*, r_i, r_{i+1}^*, \ldots, r_n^*)$.

\begin{definition*}
A strategy profile $\mathbf r^*$ is a Nash equilibrium for the travel allocation game if, for all $i$, $u_i(\mathbf r^*, s_i) \ge u_i ( (r_i, \mathbf r^*_{-i}), s_i)$ for every $r_i$. 
\end{definition*}

Proposition \ref{prop1} can be translated into an inequality involving player $i$'s utility function. Specifically, because it is in the best interest of player $i$ to request $s_i$ to minimize out-of-pocket expenses, then $u_i((s_i,\mathbf r_{-i}), s_i)  \ge u_i ( (r_i, \mathbf r_{-i}), s_i)$ for all $r_i$.  The proof of the proposition shows that the inequality is strict, so that $u_i((s_i,\mathbf r_{-i}), s_i) >  u_i ( (r_i, \mathbf r_{-i}), s_i)$ for $r_i \ne s_i$.  In game theory language, because of the strict inequality, it is a strictly dominant strategy for faculty member $i$ to request $r_i = s_i$.  Player $i$ has a dominant strategy if it maximizes player $i$'s utility independent of the other players' strategies. If every player has a dominant strategy and uses it, then the profile of dominant strategies is a Nash equilibrium.  The following proposition holds because requesting $r^*_i = s_i$ for each faculty member is a dominant strategy.

\begin{proposition} 
The strategy profile  $\mathbf r^* = (s_1, s_2, \dots, s_n)$, in which each faculty member $i$ requests $s_i$, is a Nash equilibrium for the travel allocation game.
\end{proposition}

Every player having a dominant strategy is a strong condition that assures a Nash equilibrium exists.  Interestingly, if we view the original allocation method and reimbursement function for the basis of a game, there is no Nash equilibrium.  As before, assume that there is not enough money budgeted to pay in full for every faculty members travel.  That means that at least one faculty member is paying positive out-of-pocket expenses.  Assuming all other faculty members' requests are fixed, then this faculty member can adjust his request, say $r_i$, to pay zero out-of-pocket expenses.  This means there can be no stability:  at least one faculty member can always do better by changing his request, which leads to the following proposition.

\begin{proposition}
There is no Nash equilibrium for the travel allocation game induced by the original allocation procedure.
\end{proposition}

A slight change how travel awards were paid out---by tying the amount reimbursed to the amount spent---literally was a game changer.  The process of constructing a game so that certain behavior is a Nash equilbirum, in this case the truthful revelation of what faculty members needed to spend, is called mechanism design.

 In 2007, economists Leonid Hurwicz, Eric Maskin, and Roger Myerson were awarded the Nobel Prize in Economics for their work in mechanism design \cite{Nobel1}.  As an example of another mechanism, a second-price or Vickrey auction awards an auctioned item to the bidder who bids the highest, but the bidder only pays the second-highest price. It is a weakly dominant strategy for each player to bid his or her true valuation of the item being auctioned.  Consequently, it is a Nash equilibrium for each player to bid his or her true value for the item.  However, there is more than one equilibrium for a second-price auction.  This is in contrast to the (new method) travel allocation game which only has one Nash equilibrium because each player has a strictly dominant strategy.  Examples of equilibria in which players do not bid their true valuation of the item being auctioned appear in \cite[p.645]{FAPP}.  William Vickrey was awarded Nobel Prize in Economics in 1996, along with James Mirrlees for their research into the economic theory of incentives under asymmetric information \cite{Nobel2}.  To round out the Nobel list, Nash (along John Harsanyi and Reinhard Selten)  was awarded the Nobel Prize in Economics in 1994 \cite{Nobel1}.

\section*{Afterward}

The mechanism design method I proposed was not the first one I suggested to the Associate Dean.  My initial suggestion was to allocate the same amount, $B/n$, to each of the $n$ faculty members. This would be fair because every faculty member would be allocated the same amount for travel.  I reasoned that faculty members would stretch their money, possibly paying to travel to more than one conference.  The Dean's office disagreed and was concerned that faculty members would spend all of the money, even if they needed less.  Fairness for the mechanism design approach is achieved by having faculty be reimbursed the same proportion instead.  The university adopted the new method.

\vfill\eject

\end{document}